\documentclass[11pt]{amsart}

\usepackage[foot]{amsaddr}
\usepackage[utf8]{inputenc} 
\usepackage[T1]{fontenc}    
\usepackage{hyperref}       
\usepackage{url}            
\usepackage{booktabs}       
\usepackage{amsfonts}       
\usepackage{nicefrac}       
\usepackage{microtype}      
\usepackage{xcolor}         

\usepackage{fullpage}

\usepackage{amsthm, thm-restate}
\newtheorem{theorem}{Theorem}
\newtheorem{lemma}{Lemma}
\newtheorem{fact}{Fact}

\newtheorem*{remark}{Remark}

\usepackage{wrapfig}

\newtheorem{definition}{Definition}
\usepackage[noend]{algpseudocode}
\usepackage{subfigure}
\usepackage{amsfonts}
\usepackage{algorithm}
\usepackage{amsmath}
\usepackage{graphicx}
\usepackage{amssymb}
\usepackage{bbm}
\usepackage{listings}
\usepackage{xcolor}
\usepackage{amsfonts}
\usepackage{amsthm, thm-restate}

\usepackage{tabularx, makecell, booktabs}

\usepackage{mathtools}
\usepackage{amsthm}
\usepackage[numbers]{natbib}

\usepackage{siunitx}

\definecolor{codegreen}{rgb}{0,0.6,0}
\definecolor{codegray}{rgb}{0.5,0.5,0.5}
\definecolor{codepurple}{rgb}{0.58,0,0.82}
\definecolor{backcolour}{rgb}{0.95,0.95,0.92}

\lstdefinestyle{mystyle}{
    backgroundcolor=\color{backcolour},   
    commentstyle=\color{codegreen},
    keywordstyle=\color{magenta},
    numberstyle=\tiny\color{codegray},
    stringstyle=\color{codepurple},
    basicstyle=\ttfamily\footnotesize,
    breakatwhitespace=false,         
    breaklines=true,                 
    captionpos=b,                    
    keepspaces=true,                 
    numbers=left,                    
    numbersep=5pt,                  
    showspaces=false,                
    showstringspaces=false,
    showtabs=false,                  
    tabsize=2
}

\usepackage{xspace}
\makeatletter
\DeclareRobustCommand\onedot{\futurelet\@let@token\@onedot}
\def\@onedot{\ifx\@let@token.\else.\null\fi\xspace}

\makeatother

\renewcommand{\paragraph}[1]{\noindent\textbf{#1}\quad} 

\usepackage{wrapfig}
\newcommand{\erclogowrapped}[1]{%
\setlength\intextsep{0pt}%
\begin{wrapfigure}[3]{r}{#1*\real{1.1}}%
\includegraphics[width=#1]{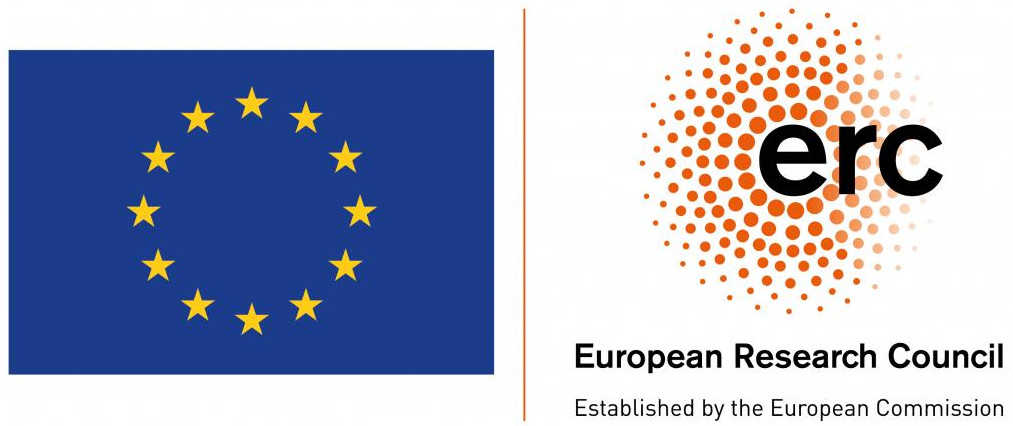}%
\end{wrapfigure}%
}

\makeatother

\author{Monika Henzinger}
\address{  
  Institute of Science and Technology (ISTA), 
  Klosterneuburg, Austria}
\email{monika.henzinger@ista.ac.at}

\author{Nikita P.  Kalinin}
\address{  
  Institute of Science and Technology (ISTA), 
  Klosterneuburg, Austria}
\email{nikita.kalinin@ist.ac.at}

\author{Jalaj Upadhyay}
\address{  
  Rutgers University, New Jersey, USA}
\email{jalaj.upadhyay@rutgers.edu}

\title{Binned Group Algebra Factorization for Differentially Private Continual Counting}

\begin{document}

\begin{abstract}

We study memory-efficient matrix factorization for differentially private counting under continual observation. While recent work by Henzinger and Upadhyay 2024 \cite{group_algebra} introduced a factorization method with reduced error based on group algebra, its practicality in streaming settings remains limited by computational constraints. We present new structural properties of the group algebra factorization, enabling the use of a binning technique from Andersson and Pagh (2024) \cite{andersson2024streaming}. By grouping similar values in rows, the binning method reduces memory usage and running time to $\tilde O(\sqrt{n})$, where $n$ is the length of the input stream, while maintaining a low error. Our work bridges the gap between theoretical improvements in factorization accuracy and practical efficiency in large-scale private learning systems.
\end{abstract}

\maketitle

\section{Introduction}
With the expansion of data collection, protecting individuals' sensitive information has become a fundamental challenge in modern computing. Differential Privacy (DP) \cite{dwork2006differential} is a well-established framework that provides strong theoretical guarantees for protecting individual privacy. However, ensuring DP under continual observation—where data is accessed and updated over time—introduces additional challenges \cite{andersson2023smooth, andersson2024improved, bolot2013private, mcmahan2022federated, kalinin2025continual, edmonds2020power, cohen2024lower}. 

In this work, we focus on the problem of \textbf{private prefix sums} under continual observation \cite{dwork2010differential}, a fundamental challenge in privacy-preserving computations. This problem arises in a wide range of applications, including differentially private Stochastic Gradient Descent (DP-SGD) for machine learning \cite{Denisov, Kairouz, mckenna2024scaling, choquette2023multi, BandedMatrix, zhang2025locally, cyffers2024differentially}, private histogram estimation \cite{cardoso2022differentially, chan2012differentially, upadhyay2019sublinear, epasto2023differentially}, private clustering \cite{dupre2024making}, and graph analysis \cite{upadhyay2021differentially, fichtenberger2021differentially}.

A promising approach for computing prefix sums under differential privacy is to use a \textbf{Matrix Factorization Mechanism} \cite{li2015matrix}. This method represents the prefix sum computation as a matrix-vector multiplication: $Mx$, where $M = (\mathbbm{1}_{j \le i}) \in \mathbb{R}^{n\times n}$ is a lower triangular matrix of coefficients and $x \in \mathbb{R}^{n}$ is the private data. To ensure differential privacy, we first factorize $M$ into the product of two matrices:  $M^{n \times n} = L^{n\times m} \times R^{m \times n}$.
Using this decomposition, we compute the intermediate product $Rx$, add appropriately scaled Gaussian noise $z \in \mathbb{R}^{m}$, and then apply a post-processing step using matrix $L$:  $\widehat{Mx} = L (Rx + z) = Mx + Lz$.
A common way to measure the introduced error is through the expected mean square error or the maximum expected square error:
\begin{equation}
    \sqrt{\mathbb{E}\|\widehat{Mx} - Mx\|_F^2 / n} \quad \text{(MeanSE)}, \qquad \max_i \sqrt{\mathbb{E}[(\widehat{Mx} - Mx)_i^2]} \quad \text{(MaxSE)}.
\end{equation} Matrix factorization allows us to achieve better accuracy than trivial factorizations, where noise is added directly to $x$ or $Mx$ via the Gaussian mechanism, while maintaining the same level of privacy.
However, one of the main challenges is that the optimal factorization of $M$ is unknown, which motivates the study of efficient factorizations that minimize error while maintaining computational feasibility.

The problem of optimal matrix factorization has been extensively studied in recent years \cite{choquette2023multi, CNPBINDPL, Henzinger, Kairouz, kalinin2024, dvijotham2024efficient, Denisov, fichtenberger2023constant, mcmahan2024hassle}. The lowest currently known factorization error has been achieved using \textbf{Group Algebra Factorization} \cite{group_algebra}, which decomposes $M$ into two non-square matrices of sizes $n \times 2n$ and $2n \times n$ with potentially complex values given by polynomials of the powers of the $2n$-th root of unity. Formally, the group algebra factorization $M^{n \times n} = L^{n \times 2n} \times R^{2n \times n}$, is given by $L_{i,j} = R_{i,j} = b_f(\omega^{j - i})$, with an appropriate range of indices $i,j$ for each matrix, for
    \begin{equation}
        b_f(x) = \frac{1}{2n} \sum\limits_{l = 0}^{2n - 1} x^l \left(\sum\limits_{k = 0}^{n - 1} \omega^{kl} \right)^{1/2},
    \end{equation}
where $\omega = e^{i\pi/n}$ a $2n$-th root of unity.  To evaluate the quality of this factorization, we use the Maximum Squared Error (MaxSE) and Mean Squared Error (MeanSE), computed as:
\begin{align}
    \text{MaxSE}(L, R) = \|L\|_{2 \to \infty} \|R\|_{1\to 2} \quad \text{and} \quad \text{MeanSE}(L, R) = \frac{1}{\sqrt{n}}\|L\|_{F} \|R\|_{1\to 2},
\end{align}
where $\|.\|_{1\to 2}$ is the maximum $\ell_2$ column norm and $\|.\|_{2 \to \infty}$ is the maximum $\ell_2$ row norm.
It has been shown that the Group Algebra Factorization satisfies the following upper bound for both MaxSE and MeanSE:
\begin{equation}
    \max\big\{\text{MeanSE}(L, R),\; \text{MaxSE}(L, R) \big\} \le \frac{1}{2} + \frac{1}{2n} \sum\limits_{l = 1}^{n} \frac{1}{\sin\left(\frac{\pi(2l - 1)}{2n}\right)} < 1 + \frac{\log n}{\pi},
\end{equation} achieving the lowest currently known theoretical bound for the factorization error.
The general known lower bound is given by Matoušek et al. \cite{matouvsek2020factorization}:
\begin{equation}
    \inf_{LR=M}\min\big\{\text{MeanSE}(L, R),\;\text{MaxSE}(L, R)\big\} \ge \frac{\log (\frac{2n + 1}{3}) + 2}{\pi},
\end{equation}
resulting in a constant gap of $\approx0.49$ between this lower bound and the best known upper bound.

A fundamental limitation of the matrix factorization mechanism is its higher computational and memory cost compared to directly adding noise to the private vector $x$, which corresponds to a trivial but suboptimal factorization. Several works have explored ways to reduce this overhead \cite{dvijotham2024efficient, andersson2024streaming, kalinin2024}. In this work, we adapt a technique introduced by Andersson and  Pagh 2024 \cite{andersson2024streaming}, which improves efficiency through \textbf{binning}\footnote{Not to be confused with quantization (which approximates individual \textit{values}, not subsegments, by mapping them to discrete levels).}. The key idea is to group the structurally similar subsegments of each row of the matrix $L$ and replace them with a shared representative value, a \textit{bin}. This process decreases memory usage while preserving the essential structure of the factorization. Furthermore, it accelerates the computation of the product with a vector $z$. This acceleration is achieved by precomputing the prefix sums of $z$ or generating noise directly in the form of prefix sums, enabling us to compute the product with a binned row in time proportional to the number of bins by efficiently retrieving partial sums of $z$ in constant time.

While binning improves efficiency, it inevitably introduces some additional approximation error. Our goal is to design a binning scheme that minimizes this error, striking a balance between computational gains and accuracy with respect to the original matrix.

\subsection{Contribution}

In this work, we present several \textbf{structural properties} of group algebra factorization matrices, allowing for a more precise analysis of their values and relationships. We prove that the values $b_{f}(\omega^{k})$ are real, resulting in a \textbf{real-valued factorization} for prefix sums. Specifically, we represent the values $b_f(\omega^k)$ as an alternating sum:
\begin{equation}
    b_f(\omega^{-t}) = \frac{1}{2\sqrt{n}} + \frac{1}{\sqrt{2}} \sum_{l = 0}^\infty \left|\binom{-1/2}{t + nl}\right| (-1)^{l},
\end{equation}
for $t \in [-n, n - 1]$, where $\binom{-1/2}{t + nl} = 0$ for $t + nl < 0$. This directly implies that $L$ and $R$ are real matrices. We use this representation to prove \textbf{upper and lower bounds} for the values $b_f(\omega^{k})$ and identify the \textbf{monotonicity regions}, namely the values $b_f(\omega^{k})$ increase for $k$ from $-n + 1$ to $0$ and decrease for $k$ from $0$ to $n$. We also prove that the group algebra factorization is equivalent to the factorization:
\begin{equation}
    M = \begin{pmatrix}
        \frac{E}{2\sqrt{n}} + \frac{C}{2} & \frac{E}{2\sqrt{n}} - \frac{C}{2}
    \end{pmatrix} \times \begin{pmatrix}
        \frac{E}{2\sqrt{n}} + \frac{C}{2}\\ \frac{E}{2\sqrt{n}} - \frac{C}{2}
    \end{pmatrix},
\end{equation}
where $E$ is an $n \times n$ matrix of all ones, and
\begin{equation}
    C = (2M - E)^{1/2} = \begin{pmatrix}
        1 &-1 & \dots &-1 &-1\\
        1 &1 & \dots & -1&-1\\
        \vdots &\vdots & \ddots &\vdots &\vdots\\
        1 & 1 &\dots & 1 & 1
    \end{pmatrix}^{1/2}.
\end{equation}

Using these insights, we address a key limitation of the group algebra factorization: \textbf{memory efficiency}. The binning method introduced in Andersson and Pagh~\cite{andersson2024improved} cannot be directly applied to group algebra factorization, as the matrices $L$ and $R$ are non-square, contain negative values, and do not satisfy the monotone ratio condition required in their work. Therefore, we propose a novel adaptation of the binning approach tailored for group algebra factorization. Specifically, we prove that group algebra  factorization can be made memory- and time-efficient, which is stated by our main theorem:

\begin{restatable}{theorem}{maintheorem}(Main Theorem)
\label{thm:main_theorem}
    Given a group algebra factorization $M = L \times R$, where $L = (L_1, L_2) \in \mathbb{R}^{n\times 2n}$ and $R = (R_1, R_2)^T \in \mathbb{R}^{2n \times n}$, we can find a binned matrix $\hat{L} = (\hat{L}_1, \hat{L}_2)$ and a corresponding matrix $\hat{R} = (\hat{L}_1^{-1}L_1R_1, \hat{L}_2^{-1}L_2R_2)^T$ such that, for any $0 < \zeta \le 1$,
    \begin{equation}
    \text{MeanSE}(\hat{L}, \hat{R}) \le \text{MeanSE}(L, R) (1 + \zeta) \quad \text{and} \quad \text{MaxSE}(\hat{L}, \hat{R}) \le \text{MaxSE}(L, R) (1 + \zeta).
    \end{equation}

   Moreover, the product of the matrix $\hat{L}$ with a vector $z \in \mathbb{R}^{2n}$ can be computed in time $O_{\zeta}(\sqrt{n}(\log n)^{3/2})$ per row. The factorization matrix $\hat{L}$  requires $O_{\zeta}(\sqrt{n}(\log n)^{3/2})$ memory for storage and can be constructed efficiently in $O(n)$ time.
\end{restatable}

\begin{remark}[Difference from Related Works]
Our work builds on the combination of Group Algebra factorization \cite{group_algebra} and Continual Counting via Binning \cite{andersson2024streaming}. The difference from the former is that we make it memory- and time-efficient, while the difference from the latter is that they use a Square Root factorization, which results in worse accuracy compared to Group Algebra factorization.
\end{remark}

\section{Preliminaries}

Recent work by Henzinger and Upadhyay 2024 \cite{group_algebra} presented the following factorization of the matrix $M = L \times R$, where: 
\begin{align}
\label{eq:group_algebra_factorization}
    L =  \begin{pmatrix}
    b_f(\omega^0) &  \dots & b_f(\omega^{2n - 1})\\
    \vdots  & \ddots & \vdots\\
    b_f(\omega^{-n + 1})  & \dots & b_f(\omega^{n})\\
\end{pmatrix} \in \mathbb{C}^{n \times 2n}, \quad 
R = \begin{pmatrix}
    b_f(\omega^0)  & \dots & b_f(\omega^{n - 1})\\
    \vdots  & \ddots & \vdots\\
    b_f(\omega^{-2n + 1})  & \dots & b_f(\omega^{-n}) \\
\end{pmatrix} \in \mathbb{C}^{2n \times n},
\end{align}
with the values $b_f(\omega^k)$ given by the function: 
\begin{equation}
        b_f(x) = \frac{1}{2n} \sum\limits_{l = 0}^{2n - 1} x^l \left(\sum\limits_{k = 0}^{n - 1} \omega^{kl} \right)^{1/2},
    \end{equation}
where $\omega = e^{i\pi/n}$ a $2n$-th root of unity. The values $b_f(\omega^{k})$ are $2n$-periodic, which makes the matrices $L$ and $R$ row- and column-\textbf{circulant}, respectively. 
The factorization quality is expressed by the following theorem:

\begin{theorem}[Theorem 1.1 from \cite{group_algebra}]
    The group algebra factorization $M^{n \times n} = L^{n \times 2n} \times R^{2n \times n}$, where $L_{i,j} = R_{i,j} = b_f(\omega^{j - i})$, satisfies  
    \begin{equation}
        \text{MaxSE}(L, R) \le \frac{1}{2} + \frac{1}{2n} \sum\limits_{l = 1}^{n} \frac{1}{\sin\left(\frac{\pi(2l - 1)}{2n}\right)} < 1 + \frac{\log n}{\pi}.
    \end{equation}
\end{theorem}

Later, we relate this factorization to the Square Root factorization \cite{Henzinger}, given by \( M^{n\times n} = C^{n\times n} \times C^{n\times n} \), where  

\begin{equation}
    C := M^{1/2} =  \begin{pmatrix}
        1 &0 & \dots & 0\\
        f_1 & 1 & \dots & 0\\
        \vdots & \vdots &\ddots & \vdots\\
        f_{n - 1} & f_{n - 2} & \dots & 1
    \end{pmatrix},
\end{equation}  

with \( f_k = |\binom{-1/2}{k}| \). The coefficients \( f_k \) have the generating function \( \sum_{k=0}^{\infty} f_k x^k = (1 - x)^{-1/2} \) and satisfy the bounds \cite{dvijotham2024efficient}:

\begin{equation}
    \frac{1}{\sqrt{\pi(k + 1)}}\le f_k \le \frac{1}{\sqrt{\pi k}}.
\end{equation}  

For binning as introduced by Andersson and  Pagh \cite{andersson2024streaming}, we need the following definitions, adjusted for non-square matrices.

\begin{definition}[$(\eta, \mu)$-perturbation]
For matrices $L, P \in \mathbb{R}^{n \times m}$ and $\eta, \mu \geq 0$, $L + P$ is an $(\eta, \mu)$-perturbation of $L$ if $|P_{i,j}| \leq \eta |L_{i,j}| + \mu$ for all $(i, j) \in [n] \times [m]$.
\end{definition}

We consider a more flexible definition of the binned matrix than that of \cite{andersson2024streaming}, without imposing restrictions on how partitions vary between rows:  
\begin{definition}[Binned Matrix]  
A \textit{binned matrix} is a matrix in which each row is assigned a binning, meaning that the elements of a row are partitioned into $k \geq 1$ disjoint subsegments, such that the entries of each have the same value.  
\end{definition}

\section{Results}
\subsection{Group Algebra Properties}

In this subsection, we establish several key properties of the group algebra factorization. The results build on each other, leading to insights into the structure and behavior of the factorization. Proofs for all lemmas are provided in the Appendix.

We begin by showing that the factorization is real-valued.

\begin{lemma}
[Factorization is real-valued]
\label{lem:b_f_omega_is_real}
The values
\begin{equation}
b_f(\omega^t) = \frac{1}{2n} \sum_{l = 0}^{2n - 1} \omega^{tl} \left( \sum_{k = 0}^{n - 1} \omega^{kl} \right)^{1/2} \in \mathbb{R}
\end{equation}
are real for all $t \in \mathbb{Z}$, where $\omega = e^{i\pi/n}$ is a $2n$-th root of unity.
\end{lemma}

The proof follows from summing values for $l$ and $2n - 1 - l$, with imaginary components canceling each other. Next, we derive an explicit decomposition for $b_f(\omega^{-t})$.

\begin{lemma}
[Closed-form expression of factors]
\label{lem:b_f_decomposition}
The values $b_f(\omega^{-t})$ are given by:
\begin{equation*}
b_f(\omega^{-t}) = \frac{1}{2\sqrt{n}} + \frac{1}{\sqrt{2}} \sum_{l = 0}^\infty \left|\binom{-1/2}{t + nl}\right| (-1)^{l},
\end{equation*}
for $t \in [-n, n - 1]$, where $\binom{-1/2}{t + nl} = 0$ for $t + nl < 0$.
\end{lemma}

%
Using this decomposition, we identify the monotonicity regions of the coefficient $b_f(\omega^{-t})$.

\begin{lemma}
[Monotonicity Property]
\label{lem:b_f_decreasing}
The values $b_f(\omega^{-t})$ increase for $t$ from $-n$ to $0$ and decrease for $t$ from $0$ to $n - 1$. By periodicity, this implies that $b_f(\omega^{-t})$ also increases for $t$ from $n$ to $2n$ and decreases for $t$ from $-2n$ to $-n - 1$.
\end{lemma}

We further prove upper and lower bounds for $b_f(\omega^{-t})$.

\begin{lemma}
\label{lem:b_f_bounds}
The following inequalities hold:
\begin{equation*}
\frac{f_t}{\sqrt{2}} - \frac{f_{t + n}}{\sqrt{2}} + \frac{4}{7\sqrt{n}} - \frac{1}{8n^{3/2}} \leq b_f(\omega^{-t}) \leq \frac{f_t}{\sqrt{2}} - \frac{f_{t + n}}{\sqrt{2}} + \frac{6}{7\sqrt{n}} + \frac{3}{8n^{3/2}},
\end{equation*}
for $t \in [-n, n - 1]$, where $f_t = 0$ for $t < 0$.
\end{lemma}

Combining monotonicity with the bounds on the values $b_f(\omega^{-t})$, we formulate the following lemma:  

\begin{lemma}
\label{lem:L_last_row}
The values in the last row of the matrix $L$, defined in \eqref{eq:group_algebra_factorization}, are given by  
\begin{equation}  
    \big[b_f(\omega^{-n+1}), \dots, b_f(\omega^{0}), b_f(\omega^{1}), \dots, b_f(\omega^{n})\big].
\end{equation}  
These values satisfy $0 < b_f(\omega^{-n+1}) \leq \cdots \leq b_f(\omega^{0}) \leq 1$
for the first $n$ values, and  $1 \geq b_f(\omega^{1}) \geq \cdots \geq b_f(\omega^{n}) \geq -1$ for the next $n$ values.  
\end{lemma}

Next, we connect group algebra-based factorization to the square root of a specific type of matrix.

\begin{lemma}
\label{lem:b_f_matrix}
The decomposition $M = L \times R$, as defined in \eqref{eq:group_algebra_factorization}, is equivalent to:
\begin{equation}
M = \begin{pmatrix}
\frac{E}{2\sqrt{n}} + \frac{C}{2}  & \frac{E}{2\sqrt{n}} - \frac{C}{2}
\end{pmatrix} \times \begin{pmatrix}
\frac{E}{2\sqrt{n}} + \frac{C}{2}\ \frac{E}{2\sqrt{n}} - \frac{C}{2}
\end{pmatrix},
\end{equation}
where $E$ is an $n \times n$ matrix of all ones, and
\begin{equation}
C = (2M - E)^{1/2} = \begin{pmatrix}
1 &-1 & \dots &-1 &-1\\
1  &1 & \dots & -1&-1\\
\vdots &\vdots & \ddots &\vdots &\vdots\\
1 & 1 &\dots & 1 & 1
\end{pmatrix}^{1/2}.
\end{equation}
\end{lemma}

Finally, we refine the analysis of the row and column norms of the factorization matrices $L$ and $R$ to show that the previously established upper bounds are indeed tight.

\begin{lemma}
\label{lem:group_algebra_matrix_norms}
The maximum row norm of the matrix $L$ is equal to the maximum column norm of the matrix $R$, and they are given by:
\begin{equation}
\|R\|^2_{1 \to 2} = \|L\|^2_{2 \to \infty} = \frac{1}{2} + \frac{1}{2n} \sum\limits_{l = 1}^{n} \frac{1}{\sin\left(\frac{\pi}{2n}(2l - 1)\right)} < 1 + \frac{\log n}{\pi}.
\end{equation}
\end{lemma}

This result confirms that the norms match the previously established upper bounds, thereby validating the accuracy of the factorization error analysis given in \cite{group_algebra}.

\subsection{Sparsification via Binning}

In this subsection, we provide several technical lemmas that lead to the proof of our main theorem~\ref{thm:main_theorem}. We first establish an upper bound on the minimal number of bins required to get an $(\eta,\mu)$-perturbation of a sequence of decreasing positive values:

\begin{lemma}[Binning for decreasing sequence]
\label{lem:bin_algorithm}
    For any decreasing sequence $1 \ge a_1 \ge a_2 \ge \cdots \ge a_n > 0$, we can efficiently in $O(n)$ time find a $(\eta, \mu)$-perturbation with at most $\frac{\log (1/\mu)}{\log (1 + 2\eta)} + 1$ bins.
\end{lemma}

\begin{proof}
We use the following binning algorithm. First, we collect all values smaller than $\mu$. Since the sequence is decreasing, these values form a contiguous segment. For the remaining values, we iteratively create bins from left to right. Specifically, for each new segment $[a_k, a_t]$, we choose the largest $t$ such that $a_t \geq \frac{a_k}{1 + 2\eta}$.
We then assign the representative value $\frac{a_k + a_t}{2}$ to this segment. This procedure runs efficiently in $O(n)$ time.

To bound the number of bins, we observe that each binning step reduces the value by at least a factor of $1 + 2\eta$ but never below $\mu$. Thus, the total number of bins is at most $\frac{\log (1/\mu)}{\log (1 + 2\eta)}$,  plus one additional bin for values smaller than $\mu$.

The resulting partition satisfies the $(\eta, \mu)$-perturbation condition. For each segment $[a_k, a_t]$ with values larger than $\mu$ (otherwise, the condition holds trivially), we verify that
\begin{equation}
    \bigg|a_i - \frac{a_k + a_t}{2}\bigg| \leq \eta a_i.
\end{equation}
Since the sequence is decreasing, it suffices to check the boundary values $a_k$ and $a_t$, leading to the conditions
\begin{equation}
    \frac{a_k - a_t}{2} \leq \eta a_t \qquad \text{and} \qquad  \frac{a_k - a_t}{2} \leq \eta a_k.
\end{equation}
The first condition is stronger and requires that $a_t \geq \frac{a_k}{1 + 2\eta}$, which is guaranteed by the algorithm. This completes the proof.
\end{proof}

Andersson and Pagh (2024)~\cite{andersson2024streaming} achieved a $(\frac{1}{c^2} - 1, \mu)$-perturbation with $O\big(\frac{\log (1/\mu)}{\log (1/c)}\big)$ bins. However, substituting $c = \frac{1}{\sqrt{1 + \eta}}$ yields the same asymptotic bound, differing only by a constant factor. Next, we state a general perturbation theorem for non-square matrices.

\begin{theorem}
\label{thm:matrix_perturbation}
Given a factorization $L \times R$, where $L \in \mathbb{R}^{n \times 2n}$ and $R \in \mathbb{R}^{2n\times n}$, with invertible submatrices $L = (L_1, L_2)$ and $R = (R_1, R_2)$, we consider a $(\eta, \mu)$-perturbation of $L$, given by the matrix $P = (P_1, P_2)$, such that $L_1 + P_1$ and $L_2 + P_2$ are invertible.   For the perturbed factorization $\hat{L} = L + P$ and  $\hat{R} = \left( (L_1 + P_1)^{-1}L_1R_1, (L_2 + P_2)^{-1}L_2R_2 \right)$, for any $0 < \zeta \leq 1$ and $i \in \{1,2\}$, where $\eta$ and $\mu$ are given by  
\begin{equation}
    \eta = \frac{\zeta}{17\cdot \max_{i \in \{1, 2\}} (\|L_i\|_{F} \|L_i^{-1}\|_{2})}, \quad 
    \mu = \frac{\zeta}{17n\cdot  \max_{i \in \{1, 2\}} (\|L_i^{-1}\|_{2})},
\end{equation}
we have:
\begin{equation*}
    \text{MeanSE}(\hat{L}, \hat{R}) \le \text{MeanSE}(L, R) (1 + \zeta) \qquad \text{and} \qquad \text{MaxSE}(\hat{L}, \hat{R}) \le \text{MaxSE}(L, R) (1 + \zeta)
\end{equation*}
\end{theorem}

The perturbation in the previous theorem is expressed in terms of the submatrices  
$L^{n \times 2n} = (L_1^{n \times n}, L_2^{n \times n})$ and their inverses.  
While the Frobenius norm of $L_1$ and $L_2$ can be easily bounded by the Frobenius norm of the whole matrix $L$, bounding the spectral norm of the inverse matrices is a real challenge.  
We will prove that there is a constant bound on their spectral norms by the following lemma:

\begin{lemma}
\label{lem:inv_norm_bound}
For a group algebra factorization \eqref{eq:group_algebra_factorization}, given by $M = L \times R$, the submatrices $L^{n \times 2n} = (L_1^{n \times n}, L_2^{n \times n})$ satisfy  
\begin{equation}
    \max\big\{\|L_1^{-1}\|_2, \|L_2^{-1}\|_2\big\} \leq 250,
\end{equation}
where $\|.\|_2$ denotes the spectral norm.
\end{lemma}

We combine the algorithm from Lemma \ref{lem:bin_algorithm}, the general approximation guarantees for non-square matrices from Theorem \ref{thm:matrix_perturbation}, and a bound on the spectral norm of the inverse submatrices given by Lemma \ref{lem:inv_norm_bound} into the final technical theorem.

\begin{theorem}
    Given a Group Algebra Factorization $M = L \times R$, where $L = (L_1, L_2) \in \mathbb{R}^{n\times 2n}$ and $R = (R_1, R_2)^T \in \mathbb{R}^{2n \times n}$, we can find a binned matrix $\hat{L} = (\hat{L}_1, \hat{L}_2)$ and a corresponding matrix $\hat{R} = (\hat{L}_1^{-1}L_1R_1, \hat{L}_2^{-1}L_2R_2)^T$ such that, for any $0 < \zeta \le 1$, each row of $\hat{L}$ has no more than $O_{\zeta}(\sqrt{n}(\log n)^{3/2})$ bins, and
    \begin{equation*}
    \text{MeanSE}(\hat{L}, \hat{R}) \le \text{MeanSE}(L, R) (1 + \zeta) \qquad \text{and} \qquad \text{MaxSE}(\hat{L}, \hat{R}) \le \text{MaxSE}(L, R) (1 + \zeta)
\end{equation*}
\end{theorem}

\begin{proof}
We apply the binning algorithm from Lemma \ref{lem:bin_algorithm} to the last row of $L$ three times:
First, to the first $n$ values, which are increasing, positive, and within $[0,1]$, as stated in Lemma \ref{lem:L_last_row}.
Second, to the next $n$ values, which we split into two segments—one for positive values and one for negative values.
We then multiply the negative values by $-1$ and bin these three segments separately. The matrix $L$ is circulant; hence, the remaining rows can be obtained by shifting the last row, allowing for an efficient computation of their binning by reusing the computed binning of the last row.
    This results in at most  
    \begin{equation}
        \frac{\log (1/\mu)}{\log (1 + 2\eta)} + 1 \text{ bins.}
    \end{equation}

    Next, we apply Theorem \ref{thm:matrix_perturbation} with the parameters  
    \begin{equation}
        \eta = \frac{\zeta}{17 \max_{i \in \{1, 2\}} (\|L_i\|_{F} \|L_i^{-1}\|_{2})}, \quad 
        \mu = \frac{\zeta}{17n \max_{i \in \{1, 2\}} (\|L_i^{-1}\|_{2})}.
    \end{equation}
    This ensures the required factorization error. Substituting these values into the bin count, and using the previously proven bounds  
    \begin{equation}
        \|L_i\|_F \le \|L\|_F \le \sqrt{n}\|L\|_{2\to \infty} = O(\sqrt{n\log n}), \quad 
        \max_{i \in \{1, 2\}} (\|L_i^{-1}\|_{2}) = O(1),
    \end{equation}
     combining Lemmas \ref{lem:group_algebra_matrix_norms} and  \ref{lem:inv_norm_bound}. Thus, the total number of bins is at most  
    \begin{equation}
        O_{\zeta}(\sqrt{n}(\log n)^{3/2}),
    \end{equation}
    concluding the proof.
\end{proof}

Finally, by combining the previous theorem with the properties of the binned matrix $\hat{L}$---specifically, that we store only the binning of the last row and shift it for the remaining rows---we achieve memory efficiency. Furthermore, precomputing the prefix sums of a vector $z$ leads to time efficiency in matrix-vector multiplication, as we only need to retrieve the binned values in a constant time and shift them according to the row index. This constitutes our main theorem:

\maintheorem*

\section{Conclusion and Open Questions}

We established several structural properties of the group algebra factorization and incorporated the binning method to improve memory and time efficiency, making this factorization more practical for streaming settings.

However, several open problems remain. While we focused on prefix sum matrices, extending these techniques to other workload matrices like exponential, momentum, and sliding window matrices is an important direction for future work. Additionally, it is unknown whether group algebra factorization can achieve logarithmic error bounds, as shown in prior work for square root factorization \cite{dvijotham2024efficient}. Further research should explore the theoretical limits of group algebra factorization and optimize its efficiency for broader applications in private continual observation tasks.

\section*{Acknowledgment}
We thank Edwige Cyffers for her valuable feedback on earlier versions of this draft.
Jalaj Upadhyay's research was funded by the Rutgers Decanal Grant no. 302918 and an unrestricted gift from Google.

\erclogowrapped{5\baselineskip}Monika Henzinger:  This project has received funding from the European Research Council (ERC) under the European Union's Horizon 2020 research and innovation programme (Grant agreement No.\ 101019564) and the Austrian Science Fund (FWF) grant DOI 10.55776/Z422 and grant DOI 10.55776/I5982.

\bibliographystyle{abbrvnat}
\bibliography{lit}

\appendix

\section{Appendix}
\subsection{Proof of Lemma \ref{lem:b_f_omega_is_real}}

\begin{proof}
    To show that $b_f(\omega^t) \in \mathbb{R}$, we begin by observing the symmetry of the summand under conjugation. For each $l$, note that $\overline{\omega^{tl}} = \omega^{-tl}$ and $\overline{\sum_{k=0}^{n-1} \omega^{kl}} = \sum_{k=0}^{n-1} \omega^{-kl}$. Pairing terms with indices $l$ and $2n - l$ in the summation reveals that they are complex conjugates, as $\omega^{t(2n-l)} = \omega^{-tl}$ and the inner sum also matches its conjugate under this pairing.

    However, there is an additional consideration since we are taking the square root of the sum $\sum_{k=0}^{n-1} \omega^{kl}$, which is generally a complex number. To ensure that the square root is consistent under conjugation, we need to check that $\sum_{k=0}^{n-1} \omega^{kl}$ is never a negative real number (see $\overline{\sqrt{-1}} = -i \ne i = \sqrt{\overline{-1}})$.

    Consider $l \neq 0$. The inner sum can be written as:
    \begin{equation}
        \sum_{k=0}^{n-1} \omega^{kl} = \frac{1 - \omega^{ln}}{1 - \omega^l} = \frac{1 - (-1)^{l}}{1 - \omega^l}.
    \end{equation}
    The numerator is a positive real number, and the denominator $1 - \omega^l$ is a non-real complex number for $l \neq 0$. Therefore, the inner sum is never a negative real number, ensuring the square root is consistent under conjugation.

    Special cases $l = 0$ and $l = n$ are straightforward. When $l = 0$, the term is $\omega^{t \cdot 0} = 1$ with a purely real inner sum of $n$. For $l = n$, $\omega^{tn}$ contributes a real factor, and the inner sum is also real since $\omega^n = -1$.

    Therefore, every term in the summation is either real or part of a conjugate pair, ensuring that the full sum is real. Hence, $b_f(\omega^t) \in \mathbb{R}$.
\end{proof}

\subsection{Proof of Lemma \ref{lem:b_f_decomposition}}
\begin{proof}

Starting with the expression for $b_f(\omega^{-t})$:

\begin{equation*}
b_f(\omega^{-t}) = \frac{1}{2n} \sum_{l = 0}^{2n - 1} \omega^{-tl} \left( \sum_{k = 0}^{n - 1} \omega^{kl} \right)^{1/2} = \frac{1}{2\sqrt{n}} + \frac{1}{2n} \sum_{l = 1}^{2n - 1} \omega^{-tl} \left( \frac{1 - \omega^{nl}}{1 - \omega^l} \right)^{1/2}
\end{equation*}
where $\omega = e^{i\pi/n}$ is a $2n$-th root of unity. For even values of $l$, the term $1 - \omega^{nl}$ is equal to $0$. Therefore, the summation can be restricted to just odd values of $l$.
\begin{equation*}
b_f(\omega^{-t}) = \frac{1}{2\sqrt{n}} + \frac{1}{\sqrt{2}n} \sum_{l = 0}^{n - 1} \omega^{-t(2l + 1)} \frac{1}{(1 - \omega^{2l + 1})^{1/2}}
\end{equation*}
We can consider the power series of $(1 - x)^{-1/2}$, which has the following form:
\begin{equation*}
(1 - x)^{-1/2} = \sum_{k = 0}^\infty f_k x^k,
\end{equation*}
where $f_k = \left|\binom{-1/2}{k}\right|$ are the coefficients given by the binomial series. We substitute $\omega^{2l + 1}$ using Abel's test \footnote{\textbf{Abel's test}: If a sequence of positive real numbers $(a_n)$ is monotonically decreasing $(a_n \geq a_{n+1})$ and $\lim_{n \to \infty} a_n = 0$, then the power series $f(z) = \sum_{n=0}^\infty a_n z^n$ converges for all $|z| = 1$ except at $z = 1$.}, as the sequence $f_k$ is monotonically decreasing and satisfies $\lim\limits_{k \to \infty} f_k = 0$.

\begin{equation*}
b_f(\omega^{-t}) = \frac{1}{2\sqrt{n}} + \frac{1}{\sqrt{2}n} \sum_{l = 0}^{n - 1} \omega^{-t(2l + 1)} \sum\limits_{k = 0}^\infty f_k \omega^{k(2l + 1)}.
\end{equation*}

By changing the order of summation, we obtain:
\begin{equation*}
b_f(\omega^{-t}) = \frac{1}{2\sqrt{n}} + \frac{1}{\sqrt{2}n} \sum\limits_{k = 0}^\infty f_k \sum_{l = 0}^{n - 1} \omega^{(k-t)(2l + 1)} = \frac{1}{2\sqrt{n}} + \frac{1}{\sqrt{2}n} \sum\limits_{k = 0}^\infty f_k \omega^{k - t} \sum_{l = 0}^{n - 1} \omega^{2(k-t)l}.
\end{equation*}

If $k \equiv t \pmod{n}$, then the inner sum is $n$, and we sum over $t + nl$; otherwise, the inner sum is zero, as $1 - \omega^{2n(k - t)} = 0$. Therefore:
\begin{equation*}
b_f(\omega^{-t}) = \frac{1}{2\sqrt{n}} + \frac{1}{\sqrt{2}} \sum_{l = 0}^\infty f_{t + nl} (-1)^{l} = \frac{1}{2\sqrt{n}} + \frac{1}{\sqrt{2}} \sum_{l = 0}^\infty \left|\binom{-1/2}{t + nl}\right| (-1)^{l},
\end{equation*}
which completes the proof.
\end{proof}

\subsection{Proof of Lemma \ref{lem:b_f_decreasing}}
\begin{proof}
We compute the difference between consecutive values:
\begin{align}
b_f(\omega^{-t - 1}) - b_f(\omega^{-t}) &= \frac{1}{\sqrt{2}} \sum_{l = 0}^\infty f_{t + 1 + nl} (-1)^l - \frac{1}{\sqrt{2}} \sum_{l = 0}^\infty f_{t + nl} (-1)^l \notag \\
&= \frac{1}{\sqrt{2}} \sum_{l = 0}^\infty \big(f_{t + 1 + nl} - f_{t + nl}\big)(-1)^l,
\end{align}
where $f_k = \left| \binom{-1/2}{k} \right| = \frac{1}{4^k}\binom{2k}{k}$.

Let $g_k = f_k - f_{k + 1}$ denote the discrete derivative of the sequence $f_k$. Proving that $g_k$ is a decreasing sequence implies that the difference $b_f(\omega^{-t - 1}) - b_f(\omega^{-t})$ has the same sign as the first term of the alternating sum.

For $t \geq 0$, the first term $f_{t + 1} - f_t$ is negative, so the sequence is decreasing for $t \geq 0$. For $t = -1$, the first term $f_0 = 1$ is positive, and the next term $(f_n - f_{n - 1})(-1)$ is also positive. Thus, the alternating sum is positive. For $t < -1$, the first term is zero, and the sum starts from $(f_{t + 1 + n} - f_{t + n})(-1)$, which is positive. Therefore, the sequence increases for $t$ from $-n$ to $0$.

It remains to show that the sequence $g_k = f_k - f_{k + 1}$ is decreasing. Specifically, we need to prove the inequality:
\begin{equation}
f_k - f_{k + 1} \geq f_{k + 1} - f_{k + 2}.
\end{equation}
This is equivalent to
\begin{align}
\label{eq:f_k_second_derivative}
f_k - 2f_{k + 1} + f_{k + 2} &= \frac{1}{4^k} \binom{2k}{k} - \frac{2}{4^{k + 1}} \binom{2k + 2}{k + 1} + \frac{1}{4^{k + 2}} \binom{2k + 4}{k + 2} \notag \\
&= f_k \left(1 - \frac{1}{2} \frac{(2k + 1)(2k + 2)}{(k + 1)^2} + \frac{1}{4^2} \frac{(2k + 1)(2k + 2)(2k + 3)(2k + 4)}{(k + 1)^2 (k + 2)^2}\right) \notag \\
&= \frac{f_k}{4} \left(4 - 4 \frac{2k + 1}{k + 1} + \frac{(2k + 1)(2k + 3)}{(k + 1)(k + 2)}\right) \notag \\
&= \frac{f_k}{4(k + 1)(k + 2)} \big(4k^2 + 12k + 8 - 8k^2 - 20k - 8 + 4k^2 + 8k + 3\big) \notag \\
&= \frac{3f_k}{4(k + 1)(k + 2)} \geq 0.
\end{align}

This completes the proof.
\end{proof}

\subsection{Proof of Lemma \ref{lem:b_f_bounds}}
\begin{proof}
We begin with the decomposition of $b_f(\omega^{-t})$ from Lemma~\ref{lem:b_f_decomposition}:
\begin{equation*}
    b_f(\omega^{-t}) = \frac{1}{2\sqrt{n}} + \frac{1}{\sqrt{2}} \sum\limits_{l = 0}^{\infty} f_{t + nl} (-1)^l = \frac{1}{2\sqrt{n}} + \frac{1}{\sqrt{2}}(f_t - f_{t + n}) +  \frac{1}{\sqrt{2}} \sum\limits_{l = 1}^{\infty} \big(f_{t + 2nl} - f_{t + n(2l+1)}\big).
\end{equation*}

The bounds for $f_k$ are given by:
\begin{equation*}
    \frac{1}{\sqrt{\pi (k + 1)}} \leq f_k \leq \frac{1}{\sqrt{\pi k}}.
\end{equation*}

To bound the sum of differences, we continue as follows:
\begin{align*}
    \sum\limits_{l = 1}^\infty \big(f_{t + 2nl} - f_{t + n(2l+1)}\big) &\leq \frac{1}{\sqrt{\pi}} \sum\limits_{l = 1}^\infty \frac{1}{\sqrt{t + 2nl}} - \frac{1}{\sqrt{t + 2nl + n + 1}}\\
    &= \frac{1}{\sqrt{\pi}} \sum\limits_{l = 1}^\infty \frac{n + 1}{\sqrt{t + 2nl}\sqrt{t + 2nl + n + 1}\big(\sqrt{t + 2nl + n + 1} + \sqrt{t + 2nl}\big)}\\
    &\leq \frac{1}{2\sqrt{\pi}} \sum\limits_{l = 1}^\infty \frac{n + 1}{(t + 2nl)^{3/2}} = \frac{n + 1}{2\sqrt{\pi}(2n)^{3/2}} \sum\limits_{l = 1}^\infty \frac{1}{\big(\frac{t}{2n} + l\big)^{3/2}}\\
    &\leq \frac{n + 1}{4\sqrt{2\pi} n^{3/2}} \sum\limits_{l = 1}^\infty \frac{1}{(-\frac{1}{2} + l)^{3/2}} = \frac{(n + 1)(2\sqrt{2} - 1)\zeta(3/2)}{4\sqrt{2\pi} n^{3/2}}.
\end{align*}

Similarly, for the lower bound:
\begin{align*}
    \sum\limits_{l = 1}^\infty \big(f_{t + 2nl} - f_{t + n(2l+1)}\big) &\geq \frac{1}{\sqrt{\pi}} \sum\limits_{l = 1}^\infty \frac{1}{\sqrt{t + 2nl + 1}} - \frac{1}{\sqrt{t + 2nl + n}}\\
    &= \frac{1}{\sqrt{\pi}} \sum\limits_{l = 1}^\infty \frac{n - 1}{\sqrt{t + 2nl + 1}\sqrt{t + 2nl + n}\big(\sqrt{t + 2nl + n} + \sqrt{t + 2nl + 1}\big)}\\
    &\geq \frac{1}{2\sqrt{\pi}} \sum\limits_{l = 1}^\infty \frac{n - 1}{(t + 2nl + n)^{3/2}} = \frac{(n - 1)}{2\sqrt{\pi}(2n)^{3/2}} \sum\limits_{l = 1}^\infty \frac{1}{\big(\frac{t + n}{2n} + l\big)^{3/2}}\\
    &\geq \frac{n - 1}{4\sqrt{2\pi} n^{3/2}} \sum\limits_{l = 1}^\infty \frac{1}{(l + 1)^{3/2}} = \frac{(n - 1)(\zeta(3/2) - 1)}{4\sqrt{2\pi} n^{3/2}}.
\end{align*}

Substituting these bounds back into the expression for $b_f(\omega^{-t})$, we derive the following:
\begin{align*}
    b_f(\omega^{-t}) &\leq \frac{1}{2\sqrt{n}} + \frac{1}{\sqrt{2}}(f_t - f_{t + n}) +  \frac{1}{\sqrt{2}}\frac{(n + 1)\zeta(3/2)(2\sqrt{2} - 1)}{4\sqrt{2\pi} n^{3/2}}\\
    &\leq \frac{f_t}{\sqrt{2}} - \frac{f_{t + n}}{\sqrt{2}} + \frac{1}{\sqrt{n}} \left(\frac{1}{2} + \frac{\zeta(3/2)(2\sqrt{2} - 1)}{8\sqrt{\pi}}\right) + \frac{\zeta(3/2)(2\sqrt{2} - 1)}{8\sqrt{\pi}n^{3/2}}\\
    &\leq \frac{f_t}{\sqrt{2}} - \frac{f_{t + n}}{\sqrt{2}} + \frac{0.837}{\sqrt{n}} + \frac{0.337}{n^{3/2}}\\
    &\leq \frac{f_t}{\sqrt{2}} - \frac{f_{t + n}}{\sqrt{2}} +\frac{6}{7\sqrt{n}} + \frac{3}{8n^{3/2}}.
\end{align*}

For the lower bound:
\begin{align*}
    b_f(\omega^{-t}) &\geq \frac{1}{2\sqrt{n}} + \frac{1}{\sqrt{2}}(f_t - f_{t + n}) +  \frac{1}{\sqrt{2}}\frac{(n - 1)(\zeta(3/2) - 1)}{4\sqrt{2\pi} n^{3/2}}\\
    &\geq \frac{f_t}{\sqrt{2}} - \frac{f_{t + n}}{\sqrt{2}} + \frac{1}{\sqrt{n}} \left(\frac{1}{2} + \frac{\zeta(3/2) - 1}{8\sqrt{\pi}}\right) - \frac{\zeta(3/2) - 1}{8\sqrt{\pi}n^{3/2}}\\
    &\geq \frac{f_t}{\sqrt{2}} - \frac{f_{t + n}}{\sqrt{2}} + \frac{0.614}{\sqrt{n}} - \frac{0.114}{n^{3/2}}\\
    &\geq \frac{f_t}{\sqrt{2}} - \frac{f_{t + n}}{\sqrt{2}} + \frac{4}{7\sqrt{n}} - \frac{1}{8n^{3/2}},
\end{align*}
completing the proof.
\end{proof}

\subsection{Proof of Lemma \ref{lem:L_last_row}}
\begin{proof}
The monotonicity follows directly from Lemma \ref{lem:b_f_decreasing}. The remaining task is to establish four bounds:  
$b_f(\omega^{0}) \leq 1$, $b_f(\omega^{1}) \leq 1$, $b_f(\omega^{-n + 1}) \geq 0$, and $b_f(\omega^{n}) \geq -1$.  

For the first inequality, we use the bound from Lemma \ref{lem:b_f_bounds}:  
\begin{equation}
    b_f(\omega^0) \leq \frac{1}{\sqrt{2}} - \frac{f_n}{\sqrt{2}} + \frac{6}{7\sqrt{n}} + \frac{3}{8n^{3/2}} 
    \leq \frac{1}{\sqrt{2}} - \frac{1}{\sqrt{2\pi (n + 1)}} + \frac{6}{7\sqrt{n}} + \frac{3}{8n^{3/2}} 
    \leq 1,
\end{equation}
where the last inequality holds for $n \geq 5$. For $1 \leq n \leq 4$, we verify numerically that $b_f(\omega^{0}) \leq 1$.  

The second inequality follows from the first, combined with Lemma \ref{lem:b_f_decreasing}, since $b_f(\omega^0) \geq b_f(\omega^{1})$.  

The third and fourth inequalities follow trivially from Lemma \ref{lem:b_f_bounds}:  
\begin{align}
   & b_f(\omega^{-n + 1}) \geq \frac{f_{n -1}}{\sqrt{2}} - \frac{f_{2n - 1}}{\sqrt{2}} + \frac{4}{7\sqrt{n}} - \frac{1}{8n^{3/2}} 
   \geq \frac{4}{7\sqrt{n}} - \frac{1}{8n^{3/2}} > 0,\\
   & b_f(\omega^{n}) \geq -\frac{1}{\sqrt{2}} + \frac{4}{7\sqrt{n}} - \frac{1}{8n^{3/2}} 
   \geq -\frac{1}{\sqrt{2}} \geq -1.
\end{align}
This concludes the proof.
\end{proof}

\subsection{Proof of Lemma \ref{lem:b_f_matrix}}
\begin{proof}
    We start by splitting the matrices $L$ and $R$ into four $n\times n$ matrices: $L_1, L_2, R_1$, and $R_2$. We first notice that $L_1 = R_1$ and $L_2 = R_2$:

    \begin{equation}
        L_1 = R_1 = \begin{pmatrix}
    b_f(\omega^0) & b_f(\omega^{1}) & \dots & b_f(\omega^{n - 1})\\
    b_f(\omega^{-1}) & b_f(\omega^{0}) & \dots & b_f(\omega^{n - 2})\\
    \vdots & \vdots & \ddots & \vdots\\
    b_f(\omega^{-n + 1}) & b_f(\omega^{1}) & \dots & b_f(\omega^{0})\\
\end{pmatrix}.
    \end{equation}

    \begin{equation*}
        L_2 = \begin{pmatrix}
    b_f(\omega^n) & b_f(\omega^{n + 1}) & \dots & b_f(\omega^{2n - 1})\\
    b_f(\omega^{n - 1}) & b_f(\omega^{n}) & \dots & b_f(\omega^{2n - 2})\\
    \vdots & \vdots & \ddots & \vdots\\
    b_f(\omega^{1}) & b_f(\omega^{2}) & \dots & b_f(\omega^{n})\\
\end{pmatrix}, \quad R_2 = \begin{pmatrix}
    b_f(\omega^{-n}) & b_f(\omega^{-n + 1}) & \dots & b_f(\omega^{- 1})\\
    b_f(\omega^{-n - 1}) & b_f(\omega^{-n}) & \dots & b_f(\omega^{- 2})\\
    \vdots & \vdots & \ddots & \vdots\\
    b_f(\omega^{-2n + 1}) & b_f(\omega^{-2n + 2}) & \dots & b_f(\omega^{-n})\\
\end{pmatrix}.
    \end{equation*}

The values $b_f(\omega^{n + k}) = b_f(\omega^{-n + k})$ for all $k$ since $\omega$ is a $2n$-th root of unity and therefore periodic. This implies that $L_2$ and $R_2$ are equal. 

Next, we establish a connection between the matrices $R_1$ and $R_2$ using the following identity, which we prove for all $k$:

\begin{equation}
\label{eq:b_f_n_plus_k}
    b_f(\omega^{n + k}) - \frac{1}{2\sqrt{n}} = - \left(b_f(\omega^{k}) - \frac{1}{2\sqrt{n}}\right). 
\end{equation}

Using the definition of $b_f(\omega^{k})$, we show that

\begin{align}
    b_f(\omega^{n + k}) + b_f(\omega^{k}) &= \frac{1}{2n} \sum_{l = 0}^{2n - 1} \omega^{(n+k)l} \left( \sum_{k = 0}^{n - 1} \omega^{kl} \right)^{1/2} + \frac{1}{2n} \sum_{l = 0}^{2n - 1} \omega^{kl} \left( \sum_{k = 0}^{n - 1} \omega^{kl} \right)^{1/2}\\
    &= \frac{1}{2n} \sum_{l = 0}^{2n - 1} (\omega^{(n+k)l} + \omega^{kl}) \left( \sum_{k = 0}^{n - 1} \omega^{kl} \right)^{1/2}\\ 
    &= \frac{1}{\sqrt{n}} + \frac{1}{2n} \sum_{l = 1}^{2n - 1} \omega^{kl}(1 + (-1)^{l}) \left( \frac{1 - (-1)^{l}}{1 - \omega^{l}} \right)^{1/2} =  \frac{1}{\sqrt{n}},
\end{align}
concluding the proof of equality \eqref{eq:b_f_n_plus_k}. Subtracting $\frac{E}{2\sqrt{n}}$ from both $R_1$ and $R_2$ shows that they sum to zero, meaning we can define the matrix $C$ such that

\begin{equation}
    R_1 = \frac{E}{2\sqrt{n}} + \frac{C}{2}, \quad \text{and} \quad R_2 = \frac{E}{2\sqrt{n}} - \frac{C}{2},
\end{equation}
which concludes the first part of the lemma. The second part follows by computing $C$ from the factorization:

\begin{equation}
    M = \left( \frac{E}{2\sqrt{n}} + \frac{C}{2}\right)^2 + \left( \frac{E}{2\sqrt{n}} - \frac{C}{2}\right)^2 = \frac{C^2}{2} + \frac{E}{2},
\end{equation}
where for the second equality we used the property $E^2 = nE$. From this it follows that $C = (2M - E)^{1/2}$. 

As a corollary, this representation implies Lemma \ref{lem:b_f_decomposition}. To show it we first need to introduce an auxiliary matrix $A$:
\begin{equation}
    (A_{i,j}) = \begin{cases}
        1, & \text{if } i - j = 1,\\
        -1, & \text{if } i = 0, j = n-1,\\
        0, & \text{otherwise}.
    \end{cases}
\end{equation}

It follows that:
\begin{equation}
    2M - E = 2 (I - A)^{-1},
\end{equation}
which we can verify by direct matrix multiplication:
\begin{equation}
    \begin{pmatrix}
        1 &-1 & \dots &-1 &-1\\
        1  &1 & \dots & -1&-1\\
        \vdots &\vdots & \ddots &\vdots &\vdots\\
        1 & 1 &\dots & 1 & 1
    \end{pmatrix} \times \begin{pmatrix}
        1 &0 & \dots &0 &1\\
        -1  &1 & \dots & 0&0\\
        \vdots &\vdots & \ddots &\vdots &\vdots\\
        0 & 0 &\dots & -1 & 1
    \end{pmatrix} = 2I.
\end{equation}

Thus, we express $C$ as:
\begin{equation}
\label{eq:C_matrix_series}
    C = \sqrt{2} (I - A)^{-1/2} = \sqrt{2} \sum\limits_{l = 0}^{\infty} f_k A^k,
\end{equation}
which, when divided by $2$ and adding $\frac{E}{2\sqrt{n}}$, gives the componentwise formula from Lemma \ref{lem:b_f_decomposition}.

\end{proof}

\subsection{Proof of Lemma \ref{lem:group_algebra_matrix_norms}}

\begin{proof}
The squared norm of any column of matrix $R$ and any row of matrix $L$ can be found as
    \begin{align}
        \sum\limits_{k = 0}^{2n - 1} b_f^2(\omega^{-k}) &= \frac{1}{4n^2} \sum\limits_{k = 0}^{2n - 1} \left(\sum\limits_{l = 0}^{2n - 1} \omega^{-kl} \left(\sum\limits_{t = 0}^{n - 1}\omega^{tl}\right)^{1/2}\right)^2\\
        &=\frac{1}{4n^2} \sum\limits_{k = 0}^{2n - 1} \left(\sqrt{n} + \sqrt{2}\sum\limits_{l = 0}^{n - 1} \frac{\omega^{-k(2l + 1)}}{(1 - \omega^{2l + 1})^{1/2}} \right)^2\\
        &= \frac{1}{4n^2} \sum\limits_{k = 0}^{2n - 1} \left[n + 2\sqrt{2n}\sum\limits_{l = 0}^{n - 1} \frac{\omega^{-k(2l + 1)}}{(1 - \omega^{2l + 1})^{1/2}} + 2\left(\sum\limits_{l = 0}^{n - 1} \frac{\omega^{-k(2l + 1)}}{(1 - \omega^{2l + 1})^{1/2}}\right)^2\right]    
    \end{align}

There are three terms. The first term trivially evaluates to $\frac{1}{2}$, while the second term can be shown to be $0$: 

\begin{equation}
    \sum\limits_{k = 0}^{2n - 1} \sum\limits_{l = 0}^{n - 1} \frac{\omega^{-k(2l + 1)}}{(1 - \omega^{2l + 1})^{1/2}} =  \sum\limits_{l = 0}^{n - 1} \frac{1}{(1 - \omega^{2l + 1})^{1/2}}\sum\limits_{k = 0}^{2n - 1}\omega^{-k(2l + 1)} = 0.
\end{equation}

Therefore, 

\begin{align}
        \sum\limits_{k = 0}^{2n - 1} b_f^2(\omega^{-k}) &= \frac{1}{2} + \frac{1}{2n^2}  \sum\limits_{k = 0}^{2n - 1} \left(\sum\limits_{l = 0}^{n - 1} \frac{\omega^{-k(2l + 1)}}{(1 - \omega^{2l + 1})^{1/2}}\right)^2\\
        &= \frac{1}{2} + \frac{1}{2n^2}  \sum\limits_{k = 0}^{2n - 1} \sum\limits_{l_1 = 0}^{n - 1}\sum\limits_{l_2 = 0}^{n - 1} \frac{\omega^{-k(2l_1 + 1)}}{(1 - \omega^{2l_1 + 1})^{1/2}} \frac{\omega^{-k(2l_2 + 1)}}{(1 - \omega^{2l_2 + 1})^{1/2}}\\
        &= \frac{1}{2} + \frac{1}{2n^2}   \sum\limits_{l_1 = 0}^{n - 1}\sum\limits_{l_2 = 0}^{n - 1} \frac{1}{(1 - \omega^{2l_1 + 1})^{1/2}(1 - \omega^{2l_2 + 1})^{1/2}}\sum\limits_{k = 0}^{2n - 1} \omega^{-k(2l_1 + 2l_2 +  2)}
\end{align}

The latter sum equals $0$ as long as $2l_1 + 2l_2 + 2 \neq 2n$. Otherwise, summing over $l_1$ with $l_2 = n - 1 - l_1$, we denote the summation variable as $l$ again, which gives:

\begin{align}
        \sum\limits_{k = 0}^{2n - 1} b_f^2(\omega^{-k}) &= \frac{1}{2} + \frac{1}{n} \sum\limits_{l = 0}^{n - 1} \frac{1}{(1 - \omega^{2l + 1})^{1/2}(1 - \omega^{2(n - 1 - l) + 1})^{1/2}}\\
        &= \frac{1}{2} + \frac{1}{n} \sum\limits_{l = 0}^{n - 1} \frac{1}{(1 - \omega^{2l + 1})^{1/2}(1 - \omega^{-2l - 1})^{1/2}}\\
        &= \frac{1}{2} + \frac{1}{n} \sum\limits_{l = 0}^{n - 1} \frac{\omega^{l + 1/2}}{i(1 - \omega^{2l + 1})} = \frac{1}{2} + \frac{1}{n} \sum\limits_{l = 0}^{n - 1} \frac{1}{i(\omega^{-l - 1/2} - \omega^{l + 1/2})}\\
        &= \frac{1}{2} + \frac{1}{2n} \sum\limits_{l = 1}^{n} \frac{1}{\sin\left(\frac{\pi}{2n}(2l - 1)\right)} <1 + \frac{\log n}{\pi},
\end{align}
where the final bound has been proved by \cite{dvijotham2024efficient}, concluding the proof.

\end{proof}

\subsection{Proof of Theorem \ref{thm:matrix_perturbation}}

\begin{proof}
We follow the proof technique of \cite{andersson2024streaming}, adapting it to the case where the matrix $L$ is not square but $L \in \mathbb{R}^{n \times 2n}$, with invertible submatrices $L_1$ and $L_2$, which may contain negative elements. We begin by proving some fundamental properties of the $(\eta, \mu)$-perturbation of such a matrix.

\begin{fact}
\label{fact:pertubationbound}
\begin{align}
    &\|L + P\|_{F} \leq (1 + \eta) \|L\|_{F} + \mu n \sqrt{2}, \\
    &\|L + P\|_{2\to \infty} \leq (1 + \eta) \|L\|_{2\to \infty} + \mu \sqrt{2n}.
\end{align}    
\end{fact}

The first bound follows from  
\begin{equation}
    \|L + P\|_{F}  \leq \|L\|_{F} + \|\eta |L| + \mu J_{n \times 2n}\|_F \leq (1 + \eta) \|L\|_{F} + \mu n \sqrt{2},
\end{equation}
where $J_{n\times 2n}$ is an all-ones matrix of size $n \times 2n$, with Frobenius norm $\sqrt{2}n$. The second bound is proved almost identically, using the fact that $\|J_{n \times 2n}\|_{2\to \infty} = \sqrt{2n}$.

\begin{fact}
\label{fact:boundonPertubation}
\begin{equation}
    \|P_i\|_2 \leq \eta \|L_i\|_F + \mu n, \quad i \in \{1, 2\}.
\end{equation}    
\end{fact}

We prove it as follows:
\begin{equation}
    \|P_i\|_2 \leq \|P_i\|_{F} = \||P_i|\|_{F} \leq \|\eta |L_i| + \mu J_{n\times n}\|_F \leq \eta\|L_i\|_F + \mu n.
\end{equation}

\begin{fact}
\label{fact:rownormBound}
\begin{equation}
    \|\hat{R}\|_{1\to 2} \leq \|R\|_{1\to 2} \Big(1 + 2 \max_{i \in \{1, 2\}} (\|P_i\|_{2} \|L_i^{-1}\|_{2})\Big).
\end{equation}    
\end{fact} 

The proof differs slightly from Lemma 3 of \cite{andersson2024streaming}. We first express $\hat{R}$ in matrix form as:
\begin{equation}
    \hat{R} = \begin{pmatrix}
        (L_1 + P_1)^{-1}L_1 & \mathbf{0} \\
        \mathbf{0} & (L_2 + P_2)^{-1}L_2
    \end{pmatrix} 
    \begin{pmatrix}
        R_1 \\
        R_2
    \end{pmatrix} = G R.
\end{equation}

Then, we bound the norm as
\begin{equation}
    \|\hat{R}\|_{1\to2} = \sup_{\|x\|_{1} \leq 1} \|GRx\|_{2} \leq \|G\|_2 \sup_{\|x\|_{1} \leq 1} \|Rx\|_{2} = \|G\|_{2} \|R\|_{1\to 2}.
\end{equation}

The remaining step is to bound the spectral norm of the matrix $G$. Since $G$ is a block matrix, its maximum singular value is the maximum of the singular values of $(L_1 + P_1)^{-1}L_1$ and $(L_2 + P_2)^{-1}L_2$, so:
\begin{equation}
    \|G\|_{2} = \max \Big( \|(L_1 + P_1)^{-1}L_1\|_2, \|(L_2 + P_2)^{-1}L_2\|_2 \Big).
\end{equation}

We bound the first term, with the second term being bounded analogously:
\begin{align}
    \|(L_1 + P_1)^{-1}L_1\|_2 &= \|I_n - (L_1 + P_1)^{-1}P_1\|_2 \leq 1 + \|(L_1 + P_1)^{-1}P_1\|_2 \\
    &\leq 1 + \|L_{1}^{-1}\|_{2} \|P_1\|_2 \|(I_n + L_1^{-1}P_1)^{-1}\|_2.
\end{align}

To proceed, we need a bound $\|L_1^{-1}P_1\|_{2} \leq \frac{1}{2}$, which we prove as follows, using Fact \ref{fact:boundonPertubation}, and a specific choice of values $\eta$ and $\mu$:

\begin{equation}
    \|L_1^{-1}P_1\|_{2} \leq \|L_1^{-1}\|_2 \|P_1\|_{2} \leq \|L_1^{-1}\|_2 (\eta \|L_1\|_F + \mu n) \leq \frac{2\zeta}{17} \leq \frac{1}{2}.
\end{equation}

We then obtain the bound:
\begin{align}
    \|(I_n + L_1^{-1}P_1)^{-1}\|_2 &= \left\|I_n + \sum\limits_{k = 1}^{\infty} (L_1^{-1}P_1)^k(-1)^k \right\|_2 \leq 1 + \sum\limits_{k = 1}^{\infty} \|(L_1^{-1}P_1)^k\|_2 \\
    &\leq 1 + \sum\limits_{k = 1}^{\infty} \|(L_1^{-1}P_1)\|^k_2 \leq 2.
\end{align}

Thus, we conclude:

\begin{equation}
    \|G\|_2 = \max_{i \in \{1, 2\}} \|(L_i + P_i)^{-1}L_i\|_2 \leq 1 + 2 \max_i (\|L_{i}^{-1}\|_{2} \|P_i\|_2),
\end{equation}
which completes the proof.

\textbf{Main Proof.}  
We now proceed with the error bounds:

\begin{equation}
\label{eq:meanse}
    \frac{\text{MeanSE}(\hat{L}, \hat{R})}{\text{MeanSE}(L, R)} = \frac{\|\hat{L}\|^2_F \|\hat{R}\|^2_{1\to 2}}{\|L\|_F^2 \|R\|_{1 \to 2}^2} \leq \left(1 + 2 \max(\|P_1\|_{2} \|L_1^{-1}\|_{2}, \|P_2\|_{2} \|L_2^{-1}\|_{2})\right)^2 \times \left(1 + \eta + \frac{\mu \sqrt{2}n}{\|L\|_F}\right)^2
\end{equation}
by using Fact \ref{fact:pertubationbound} and Fact \ref{fact:rownormBound}. Then, we use the bound for $\|P_i\|_2$ from Fact \ref{fact:boundonPertubation}:
\begin{equation*}
    \frac{\text{MeanSE}(\hat{L}, \hat{R})}{\text{MeanSE}(L, R)} \leq \left(1 + 2 \max((\eta\|L_1\|_F + \mu n) \|L_1^{-1}\|_{2}, (\eta\|L_2\|_F + \mu n) \|L_2^{-1}\|_{2})\right)^2 \times \left(1 + \eta + \frac{\mu \sqrt{2}n}{\|L\|_F}\right)^2.
\end{equation*}

We can bound the maximum of two sums as the sum of maximums. For convenience, let us define  
\begin{equation}
    \psi_{L} = \max(\|L_1\|_F \|L_1^{-1}\|_2, \|L_2\|_F \|L_2^{-1}\|_2), \quad
    \chi_{L} = \max(\|L_1^{-1}\|_2,\|L_2^{-1}\|_2),
\end{equation}
then:
\begin{equation}
    \frac{\text{MeanSE}(\hat{L}, \hat{R})}{\text{MeanSE}(L, R)} \leq \left(1 + 2 \eta \psi_L + 2\mu n \chi_L\right)^2 \times \left(1 + \eta + \frac{\mu \sqrt{2}n}{\|L\|_F}\right)^2.
\end{equation}

We now substitute the values for $\eta = \frac{\zeta}{17\psi_L}$ and $\mu = \frac{\zeta}{17n\chi_L}$, which gives:

\begin{equation}
    \frac{\text{MeanSE}(\hat{L}, \hat{R})}{\text{MeanSE}(L, R)} \leq \left(1 + \frac{4\zeta}{17}\right)^2 \times \left(1 + \frac{\zeta}{17\psi_L} + \frac{\zeta\sqrt{2}}{17\chi_L\|L\|_F}\right)^2.
\end{equation}

Next, we use the inequality

\begin{equation}
    1 = \|I_n\|_2 \leq \|L_i\|_2\|L_i^{-1}\|_2 \leq \|L_i\|_{F}\|L_i^{-1}\|_2 \leq \|L\|_F \|L_i^{-1}\|_{2}
\end{equation}
to show that $\psi_L \geq 1$ and $\|L\|_F \chi_L \geq 1$. Therefore,

\begin{equation}
    \frac{\text{MeanSE}(\hat{L}, \hat{R})}{\text{MeanSE}(L, R)} \leq \left(1 + \frac{4\zeta}{17}\right)^2 \times \left(1 + \frac{\zeta(1 + \sqrt{2})}{17}\right)^2 \leq 1 + \zeta.
\end{equation}

This holds for $\zeta \leq 1.03736$, and in particular, for $\zeta \leq 1$, which concludes the proof of the bound for MeanSE.
For the MaxSE error, we have:  

\begin{equation}
     \frac{\text{MaxSE}(\hat{L}, \hat{R})}{\text{MaxSE}(L, R)} = \frac{\|\hat{L}\|^2_{2 \to \infty}\|\hat{R}\|^2_{1 \to 2}}{\|L\|^2_{2 \to \infty}\|R\|^2_{1 \to 2}} \leq \left(1 + \eta + \frac{\mu\sqrt{2n}}{\|L\|_{2\to \infty}}\right)^2 \frac{\|\hat{R}\|^2_{1 \to 2}}{\|R\|^2_{1 \to 2}}
\end{equation}
by Fact \ref{fact:pertubationbound}. We then use the bound  $\|L\|_{2 \to \infty} \geq \frac{1}{\sqrt{n}}\|L\|_F$, 
which brings us back to the same inequality \eqref{eq:meanse}, concluding the proof.

\end{proof}

\subsection{Proof of Lemma \ref{lem:inv_norm_bound}}

\begin{proof}
We begin with the matrix representations of $L_1$ and $L_2$, as stated in Lemma \ref{lem:b_f_matrix} and its corollary given by equation \eqref{eq:C_matrix_series}:
\begin{equation}
    (L_1, L_2) = \begin{pmatrix}
        \frac{E}{2\sqrt{n}} + \frac{C}{2}  & \frac{E}{2\sqrt{n}} - \frac{C}{2}
    \end{pmatrix},
\end{equation}
where $E$ is the $n \times n$ all-ones matrix, and 
\begin{equation}
    C = \sqrt{2} (I - A)^{-1/2} = \sqrt{2} \sum\limits_{l = 0}^{\infty} f_k A^k,
\end{equation}
where the matrix $A$ is defined as  
\begin{equation}
    (A_{i,j}) = \begin{cases}
        1, & \text{if } i - j = 1,\\
        -1, & \text{if } i = 0, j = n-1,\\
        0, & \text{otherwise}.
    \end{cases}
\end{equation}
To express $L_1^{-1}$ and $L_2^{-1}$, we use the Sherman–Morrison formula, since $E = ee^T$ is a rank-one perturbation of $\pm C$:
\begin{equation}
    L_i^{-1} = 2 \left(\pm C^{-1} - \frac{C^{-1}ee^TC^{-1}}{\sqrt{n} \pm e^TC^{-1}e}\right).
\end{equation}
Thus, the spectral norm is bounded by:
\begin{equation}
    \max_{i \in \{1, 2\}}\|L_i^{-1}\|_2 \le 2 \|C^{-1}\|_2 + \frac{2\|C^{-1}ee^TC^{-1}\|_2}{\min(|e^TC^{-1}e -\sqrt{n}|, |e^TC^{-1}e +\sqrt{n}|)}.
\end{equation}
To bound the spectral norm, we will bound these three terms separately. Before proceeding with the bounds, we compute the values of the matrix $C^{-1}$. as it is the main object in the aforementioned formula:
\begin{equation}
    C^{-1} = \begin{pmatrix}
        q_0 & q_{-1} & \dots & q_{-n + 1}\\
        q_1  & q_0 & \dots & q_{-n + 2}\\
        \vdots & \vdots & \ddots & \vdots\\
        q_{n -1} & q_{n - 2} & \dots & q_{0}
    \end{pmatrix},
\end{equation}

with elements $q_t$ expressed by the following formula:

\begin{equation}
    q_t = \frac{1}{\sqrt{2}}\sum\limits_{l  =0}^{\infty} \tilde{f}_{t + nl} (-1)^l
\end{equation}
where $\tilde{f}_k = (-1)^k\binom{1/2}{k} = \frac{-1}{2k - 1}f_k$. The proof of this equation follows from the following matrix series:

\begin{equation}
    C^{-1} = \frac{1}{\sqrt{2}} (I - A)^{1/2} = \frac{1}{\sqrt{2}} \sum\limits_{k = 0}^{\infty} \tilde{f}_k A^k.
\end{equation}

We now state and prove several facts.

\textbf{Fact 0.}
\begin{equation}
    \sum\limits_{t = 0}^{n} \tilde{f}_t = f_n
\end{equation}

We  recall that $\tilde{f}_k$ has the generating function $\sqrt{1 - x}$, which, when multiplied by $\frac{1}{1 - x}$, gives $\frac{1}{\sqrt{1 - x}}$, the generating function of $f_k$, implying the convolution identity.

\textbf{Fact 1.}
\begin{equation}
    \|C^{-1}\|_{2} \le 3.
\end{equation}

We use Lemma 13 from \cite{andersson2024streaming} for the full Toeplitz matrix to show that:

\begin{align}
    \|C^{-1}\|_{2} &\le |q_0| + \sum\limits_{t = 1}^{n - 1}|q_t| + \sum\limits_{t = 1}^{n - 1}|q_{-t}| \le \frac{1 - \tilde{f}_n}{\sqrt{2}} + \frac{1}{\sqrt{2}}\sum\limits_{t = 1}^{n - 1} |\tilde{f}_{t}| + \sum\limits_{t = 1}^{n - 1} |\tilde{f}_{n - t}|\\
    &\le \sqrt{2} + \sqrt{2} \sum\limits_{t = 1}^{n - 1} (-\tilde{f}_t) \le \sqrt{2} + \sqrt{2} \sum\limits_{t = 1}^{n - 1}\frac{1}{2t - 1} \frac{1}{\sqrt{\pi t}}\\
    &\le \sqrt{2} + \sqrt{\frac{2}{\pi}} \sum\limits_{t = 1}^{\infty} \frac{1}{(2t - 1)\sqrt{t}} \le 3.
\end{align}

\textbf{Fact 2.}  
\begin{equation}
    \|C^{-1}ee^TC^{-1}\|_2 = \|C^{-1}e\|_2\|C^{-T}e\|_2 \le 2 + 2\log (n)
\end{equation}

We first find the bound for the $k$-th row sum:
\begin{align}
    (C^{-1}e)_k &= \sum\limits_{t = 0}^k q_t + \sum\limits_{t = 1}^{n - 1- k} q_{-t} = \frac{1}{\sqrt{2}}\sum\limits_{t = 0}^k\sum\limits_{l = 0}^{\infty} \tilde{f}_{t + nl} (-1)^l + \frac{1}{\sqrt{2}} \sum\limits_{t = 1}^{n - 1 - k} \sum\limits_{l = 1}^{\infty} \tilde{f}_{-t + nl} (-1)^l \\
    &=\frac{1}{\sqrt{2}}\sum\limits_{t = 0}^k\sum\limits_{l = 0}^{\infty} \tilde{f}_{t + nl} (-1)^l + \frac{1}{\sqrt{2}} \sum\limits_{t = k + 1}^{n - 1} \sum\limits_{l = 1}^{\infty} \tilde{f}_{t + n(l - 1)} (-1)^l\\
    &= \frac{1}{\sqrt{2}}\sum\limits_{l = 0}^\infty(-1)^l \left[\sum\limits_{t = 0}^{k} \tilde{f}_{t + nl} - \sum\limits_{t = k + 1}^{n - 1} \tilde{f}_{t + nl} \right].
\end{align}

We now use Fact 0 for the partial sums of $\tilde{f}_t$, therefore:
\begin{align}
    (C^{-1}e)_k &= \frac{1}{\sqrt{2}}\sum\limits_{l = 0}^{\infty} (-1)^l (2f_{nl + k} - f_{-1 + nl} - f_{n - 1 + nl})= \sqrt{2}\sum\limits_{l = 0}^{\infty}(-1)^l f_{k + nl} \le \sqrt{2} f_k.
\end{align}

Analogously, the $k$-th column sum is given by:
\begin{align}
    (C^{-T}e)_k &= \sum\limits_{t = 1}^k q_{-t} + \sum\limits_{t = 0}^{n - 1- k} q_{t} = \frac{1}{\sqrt{2}}\sum\limits_{t = 1}^k\sum\limits_{l = 1}^{\infty} \tilde{f}_{-t + nl} (-1)^l + \frac{1}{\sqrt{2}} \sum\limits_{t = 0}^{n - 1 - k} \sum\limits_{l = 0}^{\infty} \tilde{f}_{t + nl} (-1)^l \\
    &= \frac{1}{\sqrt{2}}\sum\limits_{l = 0}^\infty(-1)^l \left[-\sum\limits_{t = n - k}^{n-1} \tilde{f}_{t + nl} + \sum\limits_{t = 0}^{n - 1 - k} \tilde{f}_{t + nl} \right]\\
    &= \frac{1}{\sqrt{2}}\sum\limits_{l = 0}^\infty(-1)^l(2f_{n - k - 1 + nl} - f_{n - 1 + nl} - f_{-1 + nl})\\
    &= \sqrt{2}\sum\limits_{l = 0}^{\infty} (-1)^l f_{n -k - 1 + nl} \le \sqrt{2}f_{n - k - 1}
\end{align}

Thus,
\begin{align}
    \|C^{-1}ee^TC^{-1}\|_2 &= \|C^{-1}e\|_2\|C^{-T}e\|_2 = \sqrt{\sum\limits_{k = 0}^{n - 1}(C^{-1}e)_k^2 \times \sum\limits_{k = 0}^{n - 1}(C^{-T}e)_k^2}\\
    &\le 2 \sqrt{\sum\limits_{k = 0}^{n - 1} f_k^2 \times \sum\limits_{k = 0}^{n - 1} f_{n - k - 1}^2} = 2\sum\limits_{k = 0}^{n - 1}f_k^2 \le 2 + 2\log (n).
\end{align}
For the last inequality, we used an upper bound from Lemma 7 in \cite{kalinin2024}, thus concluding the proof.

\textbf{Fact 3.} For $n \ge 2$, the following holds:
\begin{equation}
    e^T C^{-1} e \ge 1.02 \sqrt{n}.
\end{equation}

We first assume that $n \ge 40$; for smaller values, we verify numerically that the inequality holds.

\begin{align}
    e^T C^{-1}e &= n \frac{q_0}{\sqrt{2}} + \frac{1}{\sqrt{2}} \sum\limits_{t = 1}^{n - 1}(n - t)q_t + \frac{1}{\sqrt{2}}\sum\limits_{t = 1}^{n - 1} (n - t)q_{-t} \\
    &=\frac{n}{\sqrt{2}}\sum\limits_{l = 0}^{\infty} \tilde{f}_{nl} (-1)^{l} + \frac{1}{\sqrt{2}}\sum\limits_{t = 1}^{n- 1}(n - t)\sum\limits_{l = 0}^{\infty} \tilde{f}_{t + nl} (-1)^l + \frac{1}{\sqrt{2}}\sum\limits_{t = 1}^{n- 1}(n - t)\sum\limits_{l = 0}^{\infty} \tilde{f}_{-t + nl} (-1)^l \\
    &= \frac{n}{\sqrt{2}}\sum\limits_{l = 0}^{\infty} \tilde{f}_{nl} (-1)^{l} + \frac{1}{\sqrt{2}}\sum\limits_{t = 1}^{n- 1}(n - t)\sum\limits_{l = 0}^{\infty} \tilde{f}_{t + nl} (-1)^l  - \frac{1}{\sqrt{2}}\sum\limits_{t = 1}^{n- 1}t\sum\limits_{l = 0}^{\infty} \tilde{f}_{t + nl} (-1)^l \\
    &= \frac{n}{\sqrt{2}}\sum\limits_{t = 0}^{n- 1}\sum\limits_{l = 0}^{\infty} \tilde{f}_{t + nl} (-1)^l - \sqrt{2}\sum\limits_{t = 1}^{n- 1}t\sum\limits_{l = 0}^{\infty} \tilde{f}_{t + nl} (-1)^l \\
    &= \frac{n}{\sqrt{2}}\sum\limits_{l = 0}^{\infty}(-1)^l\sum\limits_{t = 0}^{n- 1} \tilde{f}_{t + nl}  - \sqrt{2}\sum\limits_{l = 0}^{\infty}(-1)^l\sum\limits_{t = 1}^{n- 1}t \tilde{f}_{t + nl}.
\end{align}

By Fact 0, we have:
\begin{equation}
    \sum\limits_{t = 0}^{n- 1} \tilde{f}_{t + nl} = f_{n - 1 + nl} - f_{-1 + nl}.
\end{equation}

Thus, substituting:
\begin{align}
     e^T C^{-1}e &= \frac{n}{\sqrt{2}}\sum\limits_{l = 0}^{\infty}(-1)^l(f_{n - 1 + nl} - f_{-1 + nl})  - \sqrt{2}\sum\limits_{l = 0}^{\infty}(-1)^l\sum\limits_{t = 1}^{n- 1}t \tilde{f}_{t + nl} \\
     &=\sqrt{2}n\sum\limits_{l = 0}^{\infty}(-1)^lf_{n - 1 + nl}  + \sqrt{2}\sum\limits_{l = 0}^{\infty}(-1)^l\sum\limits_{t = 1}^{n- 1}\frac{t}{2(t + nl) - 1} f_{t + nl} \\
     &= 2n\left(b_f(\omega^{-(n -1)}) - \frac{1}{2\sqrt{n}}\right)  + \sqrt{2}\sum\limits_{l = 0}^{\infty}(-1)^l\sum\limits_{t = 1}^{n- 1}\frac{t}{2(t + nl) - 1} f_{t + nl}.
\end{align}

For the first sum, we use the inequality from Lemma \ref{lem:b_f_bounds}:
\begin{align}
    2n\left(b_f(\omega^{-(n -1)}) - \frac{1}{2\sqrt{n}}\right) &\ge 2n \left(\frac{f_{n - 1}}{\sqrt{2}} - \frac{f_{2n - 1}}{\sqrt{n}} + \frac{0.614}{\sqrt{n}} - \frac{0.114}{n^{3/2}} - \frac{1}{2\sqrt{n}}\right) \\
    &\ge 2n \left(\frac{1}{\sqrt{2\pi n}} - \frac{1}{\sqrt{2\pi(2n - 1)}} + \frac{0.114}{\sqrt{n}} - \frac{0.114}{n^{3/2}}\right) \\
    &\ge 2\sqrt{n} \left(\frac{1}{\sqrt{2\pi}} - \sqrt{\frac{20}{79\pi}} + 0.114\right) - \frac{0.228}{\sqrt{n}} \\
    &\ge  0.458\sqrt{n} - \frac{0.228}{\sqrt{n}},
\end{align}
where we have used that $
    \frac{1}{\sqrt{2\pi(2n - 1)}} \le \sqrt{\frac{20}{79\pi n}}, \quad \text{for } n \ge 40$.
For the second sum, we recognize that it is an alternating sum with decreasing values, therefore it is lower bounded by the first term minus the second term:

\begin{align}
     \sqrt{2}\sum\limits_{l = 0}^{\infty}(-1)^l\sum\limits_{t = 1}^{n- 1}\frac{t}{2(t + nl) - 1} f_{t + nl} 
     &\ge \sqrt{2}\sum\limits_{t = 1}^{n- 1}\frac{t}{2t - 1} f_{t} - \sqrt{2}\sum\limits_{t = 1}^{n- 1}\frac{t}{2(t + n) - 1} f_{t + n}\\
     &\ge \frac{1}{\sqrt{2\pi}} \sum\limits_{t = 1}^{n - 1} \frac{1}{\sqrt{t + 1}} -\sqrt{2} \frac{n^2}{4n - 1} f_n\\
     &\ge \frac{\sqrt{2}}{\sqrt{\pi}}(\sqrt{n} - \sqrt{2}) -\frac{\sqrt{2}}{\sqrt{\pi}}\frac{40}{159} \sqrt{n} \\
     &= \frac{119}{159} \sqrt{\frac{2}{\pi}}\sqrt{n} - \frac{2}{\sqrt{\pi}} \ge 0.597 \sqrt{n} - 1.129,
\end{align}
assuming $n \ge 40$.  Combining these bounds, we obtain:

\begin{align}
    e^T C^{-1}e &\ge 0.458\sqrt{n} - \frac{0.228}{\sqrt{n}} + 0.597 \sqrt{n} - 1.129 \\
    &\ge 1.05 \sqrt{n} - \frac{0.228}{\sqrt{n}} - 1.129 \ge 1.02 \sqrt{n},
\end{align}
where the last inequality holds for $n \ge 38$. Since we have assumed $n \ge 40$, this concludes the proof of Fact 3.

\textbf{Main proof.}

We combine all the facts to obtain a bound:

\begin{align}
    \max_{i \in \{1, 2\}}\|L_i^{-1}\|_2 &\leq 2 \|C^{-1}\|_2 + \frac{2\|C^{-1}ee^T C^{-1}\|_2}{\min(|e^T C^{-1} e - \sqrt{n}|, |e^T C^{-1} e + \sqrt{n}|)} \\
    &\leq 6 + \frac{4(1 + \log n)}{0.02\sqrt{n}} \leq 6 + 200\frac{2}{\sqrt{e}} \leq 250.
\end{align}

We note that, numerically, the maximum spectral norm does not appear to exceed 19.
\end{proof}

\end{document}